\documentclass[sigconf]{aamas}

\usepackage{balance} 

\setcopyright{none}
\acmConference[none]{}
\copyrightyear{}
\acmYear{}
\acmDOI{}
\acmPrice{}
\acmISBN{}

\title[]{Maximizing the Egalitarian Welfare in Friends and Enemies Games}

\author{Edith Elkind}
\affiliation{
  \institution{Northwestern University}
  \city{Evanston}
  \country{United States}}
\email{edith.elkind@northwestern.edu}

\author{Michele Flammini}
\affiliation{
  \institution{Gran Sasso Science Institute}
  \city{L'Aquila}
  \country{Italy}}
\email{michele.flammini@gssi.it}

\author{Giovanna Varricchio*}
\affiliation{
  \institution{University of Calabria}
  \city{Rende}
  \country{Italy}}
\email{giovanna.varricchio@unical.it}

\begin{abstract}
We consider the complexity of maximizing egalitarian welfare in \emph{Friends and Enemies Games}---a subclass of hedonic games in which every agent partitions other agents into \emph{friends} and \emph{enemies}. We investigate two classic
scenarios proposed in the literature, namely, \emph{Friends Appreciation} ($\FA$) and \emph{Enemies Aversion} ($\EA$): in the former, each agent primarily cares about the number of friends in her coalition, 
breaking ties based on the number of enemies, 
while in the latter, the opposite is true. 
For $\EA$, we show that our objective is hard to approximate within $O(n^{1-\epsilon})$, for any fixed $\epsilon>0$, and provide a polynomial-time $(n-1)$-approximation. For $\FA$, we obtain an NP-hardness result and a polynomial-time approximation algorithm. Our algorithm achieves a ratio of $2-\Theta(\frac{1}{n})$ when every agent has at least two friends; however, if some agent has at most one friend, its approximation ratio deteriorates to $n/2$.
We recover the $2-\Theta(\frac{1}{n})$ approximation ratio for two important variants: when randomization is allowed and when the friendship relationship is symmetric.
Additionally, for both $\EA$ and $\FA$ we identify special cases where the optimal egalitarian partition can be computed in polynomial time.
\end{abstract}


\newcommand{\BibTeX}{\rm B\kern-.05em{\sc i\kern-.025em b}\kern-.08em\TeX}

\usepackage{amsmath}
\usepackage{amsthm}
\usepackage{enumitem}
\usepackage{cleveref}
\usepackage{marginnote}
\usepackage{todonotes}

\usepackage{tikz}
\usetikzlibrary{calc} 
\usetikzlibrary{arrows}

\usepackage{subcaption}

\usepackage{microtype}
\usepackage{varwidth}

\usepackage{thm-restate}

\newtheorem{observation}{Observation}
\newtheorem{theorem}{Theorem}
\newtheorem{lemma}{Lemma}
\newtheorem{corollary}{Corollary}

\theoremstyle{definition}
\newtheorem{example}{Example}

\newtheorem{remark}{Remark}

\newcommand{\set}[1]{\left\{#1\right\}}

\newcommand{\modulus}[1]{\left\lvert #1 \right\rvert }
\DeclareMathOperator{\expect}{\mathbb{E}}
\newcommand{\expectation}[1]{\expect\left[#1\right]}
\newcommand{\exptDistr}[2]{\expect_{#2}\left[#1\right]}

\usepackage[linesnumbered,ruled,vlined]{algorithm2e}

\newcommand{\agents}{\mathcal{N}}

\newcommand{\ESW}{\mathsf{ESW}}
\newcommand{\maxEg}{\textsc{maxESW}}
\newcommand{\FA}{\mathsf{FA}}
\newcommand{\EA}{\mathsf{EA}}
\newcommand{\outcomes}{\Pi}

\newcommand{\Gf}{G^f}
\newcommand{\Gsf}{G^{\mathit{sf}}}

\newcommand{\PC}{\textsc{PC}}
\newcommand{\minPC}{\textsc{minPC}}

\newcommand{\allF}{\textsc{allF}}
\newcommand{\oneF}{\textsc{oneF}}
\newcommand{\onlyF}{\textsc{onlyF}}
\newcommand{\weaklyC}{\text{\normalfont\ttfamily WeaklyConn}}
\newcommand{\oneFalgo}{\text{\normalfont\ttfamily OneWeaklyConn}}
\newcommand{\randAlgo}{\text{\normalfont\ttfamily RandAlgo}}

\newcommand{\OPT}{\textsf{OPT}}
\newcommand{\opt}{\textsf{opt}}
\newcommand{\instance}{\mathcal{I}}

\newcommand{\classNP}{\textsf{NP}}

\newcommand{\classP}{\textsf{P}}

\newcounter{mechanismCounter}
\crefname{mechanism}{Mechanism}{Mechanism}

\begin{document}


\pagestyle{fancy}
\fancyhead{}


\maketitle 


\section{Introduction}
Hedonic Games (HGs)~\cite{dreze1980hedonic} provide a game-theoretic framework for analyzing coalition formation among selfish agents and have been extensively studied in the literature (see, e.g., \cite{aziz2013pareto,aziz2013computing,banerjee2001core,bogomolnaia2002stability,elkind2009hedonic,elkind2020price,gairing2010computing}).
In these games, the objective is to partition a set of agents into disjoint coalitions, with each agent's satisfaction determined solely by the members of her coalition.

Different preference models give rise to different
subclasses of HGs. For instance, in \emph{additively separable} HGs (ASHGs)~\cite{bogomolnaia2002stability}, agents assign values to each other and evaluate coalitions by summing the values they assign to every other member. A particularly natural scenario is when agents classify others as either \emph{friends} or \emph{enemies}. Two canonical preference models have been studied in this setting~\cite{dimitrov2006simple}: under \emph{Friends Appreciation} ($\FA$) preferences, agents prioritize coalitions with more friends and, in case of ties, prefer those with fewer enemies; conversely, under \emph{Enemies Aversion} ($\EA$) preferences, agents prefer coalitions with the lowest number of enemies and, if the number of enemies is the same, favor those with more friends.

While most of the existing literature on HGs has focused on stability, i.e., the resilience of an outcome to individual or group deviations, another important concept in coalition formation is the {\em social welfare}, which captures the quality of a partition. Specifically, the {\em utilitarian} welfare, which is given by the sum of the agents' utilities, measures the overall happiness of the agents in a given partition. However, maximizing utilitarian welfare may come at a cost of treating some agents badly if doing so increases the total sum. In contrast, the {\em egalitarian} welfare is determined by the minimum utility among all agents. Thus, a partition that maximizes the egalitarian welfare is fairer, in the sense that the lowest utility is as high as possible.
Despite its appeal, egalitarian welfare has received limited attention in the context of hedonic games.

\subsection{Our Contribution}
We study the computation of the maximum egalitarian welfare in Friends and Enemies Games. 

For $\EA$ games, we prove that the problem is hard to approximate within $O(n^{1-\epsilon})$ for any fixed $\epsilon>0$,  where $n$ is the number of agents, and complement this hardness result with a polynomial-time algorithm achieving an $(n-1)$-approximation. For $\FA$ games, we show that maximizing egalitarian welfare is computationally hard. We then provide a polynomial-time approximation algorithm whose performance depends on the structure of the instance: it achieves an approximation of $2 -\Theta(\frac{1}{n})$, when every agent has at least two friends, and an $n/2$-approximation in the worst case. We further improve this result in two important settings, namely, (1) when randomization is allowed, and (2) when the preferences are symmetric, that is, the friendship relations are mutual. For both cases, we design algorithms attaining an approximation ratio of $2-\Theta(\tfrac{1}{n})$. Finally, we complement these results by identifying special cases of both $\EA$ games and $\FA$ games where the optimal egalitarian partition can be computed exactly in polynomial time.

\subsection{Related Work}
There is a substantial body of literature on hedonic games; we refer the interested reader to~\cite{AzizS16} for a comprehensive overview. In what follows, we focus on work that is directly related to ours.

Aziz et al.~\cite{aziz2013computing} consider
ASHGs, and show that computing an optimal  egalitarian partition is strongly \classNP-hard; moreover, verifying whether a given partition  maximizes the egalitarian welfare is co\classNP-complete. The complexity of this problem is further investigated by Hanaka et al.~\cite{hanaka2019computational}, who identify several tractable cases. Peters~\cite{peters2016graphical} obtained a related result, showing that an optimal egalitarian partition is polynomial-time computable in graphical HGs as long as the underlying graph has bounded treewidth. In addition, egalitarian welfare has been studied in online ASHGs~\cite{cohen2025egalitarianism} and in ASHGs where at most $k$ coalitions can be formed~\cite{waxman2020maximizing}.
For {\em simple fractional hedonic games} (sFHGs) with symmetric preferences, Aziz et al.~\cite{aziz2015welfare} show that maximizing the egalitarian welfare is \classNP-hard, and provide a polynomial-time $3$-approximation algorithm. Aziz et al.~also state that in instances where the underlying graph is a tree, an optimal egalitarian partition can be computed in polynomial time via dynamic programming. This idea was later formalized by Hanaka et al.~\cite{hanaka2025maximizing}, who present an efficient algorithm for bounded treewidth graphs. The egalitarian welfare has also been employed as a quality measure for evaluating stable outcomes in modified fractional HGs~\cite{monaco2019performance}.

A closely related problem is the existence of a wonderful (a.k.a. perfect) partition in HGs where agents have {\em dichotomous preferences}, that is, each agent values each coalition as either $0$ or $1$. A wonderful partition is one where each agent assigns value $1$ to her coalition. Clearly, a wonderful partition exists 
if an only if there exists a partition whose egalitarian welfare is $1$.
Establishing whether such a partition exists is, in general, \classNP-hard~\cite{peters2016complexity,aziz2013pareto}, though some tractable cases have been identified~\cite{peters2016complexity,constantinescu2023solving}. Notably, the positive result of Constantinescu et al.~\cite{constantinescu2023solving} implies that the problem of maximizing the egalitarian welfare in anonymous HGs with single-peaked preferences is polynomial-time solvable.

Friends and Enemies Games constitute a widely studied subclass of hedonic games and have been further extended to capture more complex social contexts~\cite{dimitrov2006simple,rothe2018borda,kerkmann2022altruistic,barrot2019unknown,deligkas2025balanced,chen2023Hedonic}. In this line of research, most of the focus has been on finding desirable outcomes according to various stability notions. Beyond stability, aspects related to strategyproofness have also been investigated~\cite{dimitrov2004enemies,flammini2022strategyproof,flammini2025non,klaus2023core}. From a welfare perspective, however, only the utilitarian welfare has been studied. In particular, for $\EA$ preferences, maximum  utilitarian welfare is not approximable within a factor of $O(n^{1-\epsilon})$, where $n$ is the number of agents; however, there is a polynomial-time  $O(n)$-approximation algorithm~\cite{flammini2022strategyproof}. This bound can be improved, with high probability, in scenarios where friend and enemy relations are generated according to specific probabilistic models~\cite{bullinger2025welfare}.  Even under $\FA$ preferences, the problem remains \classNP-hard, though a polynomial-time $(4+o(1))$-approximation has been obtained~\cite{flammini2025non}.

To the best of our knowledge, egalitarian welfare has not previously been considered in Friends and Enemies Games.

\section{Model and Preliminaries}
Given a positive integer $k$, we denote by $[k]$ the set $\set{1,\dots, k}$.

In hedonic games, a set of $n$ agents $\agents=\set{1,\dots,n}$ needs to be partitioned into disjoint subsets. We refer to subsets of $\agents$ as {\em coalitions}; the set of all coalitions containing an agent $i\in\agents$ is denoted by $\agents^i$.
An outcome of a hedonic game with the set of agents $\agents$ is a {\em coalition structure}, i.e., a partition $\pi=\set{C_1,\dots, C_m}$ of $\agents$ such that $\cup_{k=1}^m C_k = \agents$ and $C_k \cap C_\ell= \emptyset$ for all $k, \ell\in [m]$ with $k\neq \ell$. We denote by $\outcomes$ the set of all partitions of $\agents$.
Given an outcome $\pi\in\outcomes$, we denote by $\pi(i)$ the coalition in $\pi$ that contains agent $i$. The coalition $\agents$ is called the \emph{grand coalition};
abusing terminology, 
we also refer to the partition
$\set{\agents}$ as the grand coalition.
A coalition of size~$1$ is called a {\em singleton coalition}.   

In hedonic games, each agent $i\in\agents$ is endowed with a preference relation $\succeq_i$ over $\agents^i$. Given two coalitions $X, Y\in\agents^i$, we say that $i$ {\em weakly prefers} $X$ to $Y$ if $X\succeq_i Y$. These preferences are lifted to partitions: an agent $i\in\agents$ weakly prefers $\pi$ to $\pi'$ if $\pi(i)\succeq_i\pi'(i)$.

\subsection{Friends and Enemies Games}
In hedonic games with friends and enemies, every agent $i\in\agents$ partitions the other agents into a set of friends $F_i$ and a set of enemies $E_i$, with $F_i \cup E_i = \agents \setminus \{i\}$ and $F_i \cap E_i = \emptyset$. We write $f_i:=|F_i|$,  $e_i=|E_i|=n-f_i-1$. 

The agents' preferences 
are based on {\em Friends Appreciation} ($\FA$) when
for each $i\in\agents$ it holds that  $X\succeq_i Y$ iff
\begin{align*}
&\modulus{X\cap F_i}> \modulus{Y\cap F_i} \; \; \mbox{or}\\
&\modulus{X\cap F_i} = \modulus{Y\cap F_i} \mbox{ and }  \modulus{X\cap E_i} \leq \modulus{Y\cap E_i} \enspace,
\end{align*}
and they are based on {\em Enemies Aversion} ($\EA$) when $X\succeq_iY$ iff
\begin{align*}
&\modulus{X\cap E_i}< \modulus{Y\cap E_i} \; \; \mbox{or}\\
&\modulus{X\cap E_i} = \modulus{Y\cap E_i} \mbox{ and }  \modulus{X\cap F_i} \geq \modulus{Y\cap F_i} \enspace.
\end{align*}
In other words, under $\FA$, a coalition is preferred over another one if it contains a higher number of friends; if the number of friends is the same, the coalition with fewer enemies is preferred. On the other hand, under $\EA$, a coalition is preferred if it contains fewer enemies; if the number of enemies is the same, the coalition with more friends is preferred.

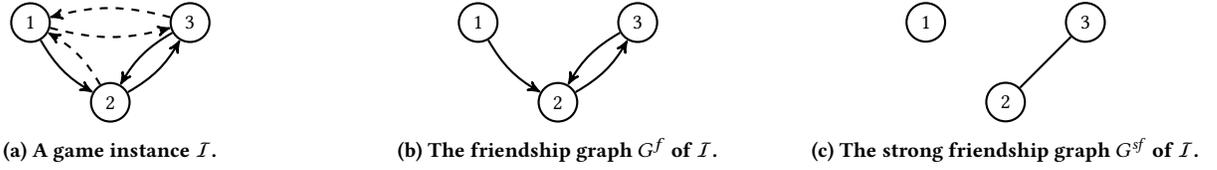
\begin{figure*}[th]\centering 
\def \variable {1.5cm}
	\begin{subfigure}{0.33\textwidth}
		\centering
		\begin{tikzpicture}[->,>=stealth',auto,node distance=\variable,
		thick,main node/.style={circle,draw,font=\sffamily\bfseries\small}]
		\node[main node] (2) {$2$};
		\node[main node] (1) [above left of=2] {$1$};
		\node[main node] (3) [above right of=2] {$3$};
		\path[every node/.style={font=\sffamily\small}]
		(1) edge [bend right =15]  node {} (2)
		(3) edge [bend right =15] node {} (2)
		(2) edge [bend right =15] node {} (3)
		(2) edge [bend right =15, dashed]  node {} (1)
		(1) edge [bend right =15, dashed]  node {} (3)
		(3) edge [bend right =15, dashed]  node {} (1)	
		;
		\end{tikzpicture}
		\caption{A game instance $\instance$.}
	\label{subfig::firstInstance}
	\end{subfigure}
	\begin{subfigure}{0.33\textwidth}
		\centering
		\begin{tikzpicture}[->,>=stealth',auto,node distance=\variable,
		thick,main node/.style={circle,draw,font=\sffamily\bfseries\small}]
		\node[main node] (2) {$2$};
		\node[main node] (1) [above left of=2] {$1$};
		\node[main node] (3) [above right of=2] {$3$};
		\path[every node/.style={font=\sffamily\small}]
		(1) edge [bend right =15]  node {} (2)
		(3) edge [bend right =15] node {} (2)
		(2) edge [bend right =15] node {} (3);
		\end{tikzpicture}
		\caption{The friendship graph~$\Gf$ of $\instance$.}
		\label{subfig::secondInstance}
	\end{subfigure}
    	\begin{subfigure}{0.33\textwidth}
		\centering
		\begin{tikzpicture}[-,>=stealth',auto,node distance=\variable,
		thick,main node/.style={circle,draw,font=\sffamily\bfseries\small}]
		\node[main node] (2) {$2$};
		\node[main node] (1) [above left of=2] {$1$};
		\node[main node] (3) [above right of=2] {$3$};
		\path[every node/.style={font=\sffamily\small}]
		(2) edge [] node {} (3);
		\end{tikzpicture}
		\caption{The strong friendship graph~$\Gsf$ of $\instance$.}
		\label{subfig::thirdInstance}
	\end{subfigure}
 \caption{A game instance $\instance$ of friends and enemies games and the corresponding graphs $\Gf$ and $\Gsf$. }
\end{figure*}
\begin{example} \label{example::instance}
Let us describe a simple instance for both $\FA$ and $\EA$. 
Note that once the friendship relationships are given, the enemy relationships are uniquely determined. Therefore, we can define an instance by specifying the friendship relations.
Let $\agents=\set{1,2,3}$ be the set of agents, and let $F_1=\set{2}$, $F_2=\set{3}$, $F_3=\set{2}$;
see \Cref{subfig::firstInstance}, where a directed edge from agent $i$ to agent $j$ represents $i$'s opinion of $j$, and solid and dashed edges represent friend and enemy relations, respectively. 
\end{example}

A well-studied class of hedonic games is {\em additively separable hedonic games (ASHGs)}, where agents assign values $v_i(j)$ to each other, the utility $u_i(C)$ an agent $i$ derives from a coalition $C\in\agents^i$ is the sum of her values for the coalition members other than herself, i.e., $u_i(C) = \sum_{j\in C\setminus\{i\}}v_i(j)$, and $C\succeq_i D$ iff $u_i(C)\ge u_i(D)$. Given a partition $\pi$, for each $i\in\agents$ we write $u_i(\pi):=u_i(\pi(i))$. 
Both $\FA$ games and $\EA$ games can be encoded as ASHGs, with value functions determined by friendship relationships.
Indeed, in the $\FA$ case, for every agent $i \in \agents$, we can set 
\begin{align*}
v_i(j)= \begin{cases}
1,& \enspace \mbox{ if } j\in F_i, \\ -\frac{1}{n},& \enspace \mbox{ if } j\in E_i \, .
\end{cases}
\end{align*}
so that the positive effect of one friend is greater than the overall negative effect of all enemies.
Similarly, in the $\EA$ case, we can set, for every agent $i \in \agents$, 
\begin{align*}
v_i(j)= \begin{cases}
1,& \enspace \mbox{ if } j\in F_i, \\ -n,& \enspace \mbox{ if } j\in E_i \, .
\end{cases}
\end{align*}

\begin{example}\label{example::computingUtilities}
Consider the instance described in \Cref{example::instance} and the partition $\pi=\set{\set{1,2,3}}$. Each agent $i\in\agents$ has one friend and one enemy in $\agents$, so 
under $\FA$ we have $u_1(\pi)= v_1(2) +v_1(3) = 1 - \frac{1}{3} = \frac{2}{3}$ and, similarly,  $u_2(\pi)=u_3(\pi)=\frac{2}{3}$, whereas under $\EA$ we obtain $u_1(\pi)= v_1(2) +v_1(3) = 1 - 3 = -2$, and also $u_2(\pi)=u_3(\pi)=-2$.

Observe that, under $\FA$, the presence of an enemy in a coalition may be tolerated as long as there is at least one friend, while, under $\EA$, the presence of an enemy is never acceptable, 
in the sense that an agent would prefer to be a singleton coalition.
\end{example}

An $\FA$ or $\EA$ instance $\instance$ is specified by a set of agents $\agents$ and the set of friendship relations $F=\set{(i,j) \,\vert\, i\in\agents, j\in F_i}$.  We denote by $F^s$ the set of mutual friendships: $F^s=\set{\set{i,j} \mid i\in F_j \wedge j\in F_i}$.

\paragraph{Graph Representation.}
In our discussion, we will make use of two different types of graph representation:
\begin{itemize}
    \item The {\em friendship graph} $\Gf=(\agents, F)$, representing all directed friendship relationships, and
    \item the {\em strong friendship graph} $\Gsf=(\agents, F^s)$, representing only mutual friendship relationships. 
\end{itemize} \Cref{subfig::secondInstance,subfig::thirdInstance} show the friendship graph and the strong friendship graph, respectively, for the instance described in \Cref{example::instance}.
When an instance is {\em symmetric}, that is, for all $i, j\in\agents$ it holds that $i\in F_j$ if and only if $j\in F_i$, we have $\Gf\equiv\Gsf$.

Throughout the paper, we may define a game instance either as $\instance = (\agents, F)$ or by directly referring to the respective  graph $\Gf$ or $\Gsf$, depending on the context. For simplicity, we will often omit $\instance$ from the notation, implicitly assuming that $\agents$ and $F$ are given.

\subsection{Egalitarian Welfare, Approximation, and Other Guarantees}
For an instance $\instance$ of an ASHG, we define the {\em egalitarian welfare} of an outcome $\pi$ as 
\[
\ESW^\instance(\pi)= \min_{i\in \agents} u_i(\pi)\, .
\]
We are interested in finding outcomes that (approximately) maximize $\ESW$. Accordingly, we define \maxEg\ as the problem of maximizing the egalitarian welfare. Given an instance $\instance$, we denote by $\opt(\instance)$ the value of an optimal solution of \maxEg\ for the instance $\instance$, and by $\OPT(\instance)$ any partition of the agents such that $\ESW^\instance(\OPT(\instance))=\opt(\instance)$. When $\instance$ is clear from the context, we simply write $\ESW$, $\opt$, and $\OPT$.

\begin{remark}\label{rem:Fneqempty}
In any instance of $\EA$ or $\FA$, if  $F_i=\emptyset$ for some $i\in\agents$, then the maximum attainable egalitarian welfare is necessarily $0$, since $i$ cannot achieve a positive utility in any coalition. In this case, partitioning the agents into singletons is optimal. Henceforth, we will always assume that $F_i\neq\emptyset$ for every $i\in\agents$.
\end{remark}

\paragraph{Approximation Ratio.}
In this work, we aim to find partitions that maximize the egalitarian welfare. Unfortunately, it turns out that \maxEg\ is \classNP-hard, both for $\EA$ and for $\FA$. For this reason, we will turn our attention to approximate solutions. Given an algorithm $\mathcal{A}$, we define its {\em approximation ratio} as follows:
\begin{align*}
    r_n(\mathcal{A})= \max_{\substack{\instance=(\agents, F):\\|\agents|=n}} \frac{\opt(\instance)}{\ESW^\instance(\mathcal{A}(\instance))}\quad\ \ \text{and}\quad\ \  r({\mathcal A})=\sup_{n\in\mathbb N}r_n({\mathcal A})\, ,
\end{align*}
where $\mathcal{A} (\instance)$ is the partition returned by $\mathcal{A}$ on input $\instance$. 

\paragraph{Guarantees for the Agents.}
In the pursuit of finding efficient algorithms that yield good approximation for \maxEg, we will simultaneously provide some guarantees to  (almost) all agents.
In particular, given an agent $i$, we say that a coalition $C$ satisfies for $i$:
\begin{itemize}
    \item \allF, if all friends of $i$ are included in $C$;
    \item \onlyF, if $C$ consists solely of friends of $i$, though not necessarily all of them;
    \item \oneF, if $C$ contains at least one friend of $i$.
\end{itemize}

Let $\Gamma$ be a guarantee that a coalition may fulfill for an agent. 
A partition $\pi$  {\em guarantees $\Gamma$ for $A\subseteq\agents$} if each coalition $C\in \pi$ guarantees $\mathcal{G}$ to all agents in $A\cap C$.
Moreover, $\pi$
is a {\em minimal-by-refinement partition guaranteeing $\Gamma$ for $A\subseteq\agents$}
if
it is not possible to split $C$ into non-empty coalitions $C_1,\dots, C_k$ so that the partition $\pi' = \pi \setminus\set{C} \cup \set{{C_1}, \dots, {C_k}}$
guarantees $\Gamma$ for $A\subseteq\agents$. If $A=\agents,$ we say that $\pi$ is a {\em minimal-by-refinement partition guaranteeing $\Gamma$}.

We now present an example to illustrate these notions.

\begin{example}
    The grand coalition $\pi=\set{\agents}$ satisfies \allF. However, it may fail \onlyF, as there may exist an agent $i$ with $F_i\neq \agents\setminus\set{i}$. Under our hypothesis that each agent has at least one friend, $\pi$ also satisfies \oneF. 
    In fact, under this hypothesis \allF\ implies \oneF.
    However, $\{\agents\}$ need not be the minimal-by-refinement partition that satisfies \allF. Indeed, if $\agents=\set{1,2,3,4}$ and $F_1=\set{2}, F_2=\set{1}, F_3=\set{4},$ and $F_4=\set{3}$,  then the unique minimal-by-refinement partition satisfying \allF\ is $\set{\set{1,2},\set{3,4}}$.
\end{example}

\section{Enemies Aversion}
When constructing (approximately) optimal partitions under $\EA$ preferences, we can restrict our attention to mutual friendship relationships only. 
Indeed, if some agent is in a coalition with an enemy, her utility---and hence $\ESW$---is negative, whereas for the coalition structure $\pi^\text{sing}$ that consists of $n$ singleton coalitions we have $\ESW(\pi^\text{sing})=0$.
Consequently, for $\ESW$ to be non-negative, an outcome $\pi$ must satisfy \onlyF\ for all agents, i.e., each coalition $C$ in $\pi$ must induce a clique in the strong friendship graph $\Gsf$. Moreover, if $\pi$ contains a singleton coalition, we have $\ESW(\pi)\le 0$, even if $\pi$ satisfies \onlyF\ for all agents. Thus, for $\ESW(\pi)$ to be positive, $\pi$ must satisfy \onlyF\ and \oneF\ for all agents, i.e., induce a partition of $\Gsf$ into cliques of size at least $2$.  

Combining these insights, we derive the following lemma.

\begin{lemma}\label{lemma:cliquePartitionPositiveESW}
    A partition $\pi\in\Pi$ satisfies $\ESW(\pi) \geq 0$ if and only if $\pi$ induces a clique partition of $\Gsf$. Furthermore, $\ESW(\pi) > 0$ if and only this clique partition contains no cliques of size~$1$.
\end{lemma}

Thus, in an optimal coalition structure, the set of agents must be partitioned into cliques. Moreover, to maximize the $\ESW$, the smallest clique in the partition must be as large as possible.

However, we will now see that this structural result does not mean that maximizing $\ESW$ is easy. 

\subsection{Hardness of Approximation}
\begin{theorem}\label{thm:inapproxEA}
Under $\EA$, for every $\epsilon> 0$, there is no polynomial-time algorithm that approximates $\maxEg$ within a factor $O(n^{1 - \epsilon})$, unless $\classP= \classNP$. The hardness result holds even for symmetric instances.
\end{theorem}

To prove the theorem, we make use of an inapproximability result for the closely related problem: {\sc Partition into Cliques}.
{
\smallskip
\begin{center}
\fbox{
\begin{varwidth}{0.95\linewidth}
\noindent {\sc Partition into Cliques (PC)}:

\smallskip
\noindent\textit{Input:} An undirected graph $G = (V, E)$ and a positive integer $K\leq \modulus{V}$.

\smallskip
\noindent\textit{Question:} Can the vertices of $G$ be partitioned into $k\leq K$ disjoint sets $V_1, \dots, V_k$ so that, for each $h\in [k]$, the subgraph $G[V_h]$ induced by $V_h$ is a clique?
\end{varwidth}
}
\end{center}
\smallskip
}

Let \minPC\ denote the corresponding minimization problem, i.e., finding the minimum number of cliques in a clique partition of $G$. Note that every graph admits a clique partition, since a single vertex forms a clique of size $1$. The problem \minPC\ cannot be approximated within a factor  $O(\modulus{V}^{1 - \epsilon})$, for every $\epsilon>0$, unless $\classP= \classNP$~\cite{zuckerman2007linear}.

\begin{proof}[Proof of \Cref{thm:inapproxEA}]
Assume there exists an $n^{1-\epsilon}$-approxima\-tion algorithm $\mathcal A$ for \maxEg\ under $\EA$,  
for some $\epsilon>0$. We will show that this assumption contradicts the known inapproximability result for the \minPC\ problem. 

Let $G = (V, E)$ be an arbitrary instance of \minPC. We now construct a corresponding $\EA$ instance $\mathcal I$ of the \maxEg\ problem. By \Cref{lemma:cliquePartitionPositiveESW}, it suffices to describe the friendship graph $\Gsf = (\agents, F^s)$ for $\mathcal I$, since only coalition structures inducing clique partitions of $\Gsf$ offer non-negative $\ESW$.
To construct our instance $\mathcal I$, we set $\agents = V\cup V'$, where $|V'|=|V|$ and $V'$ is disjoint from $V$. For each 
$\set{u, v}\in E$, we create an edge $\set{u,v} \in F^s$. Additionally, we create an edge between every pair of vertices in $V'$ and connect every vertex of $V$ to every vertex of $V'$.
Note that $\mathcal I$ satisfies $n=\modulus{\agents}=2\cdot \modulus{V}$.

Let us now relate $\opt(\instance)$, i.e., the optimal \maxEg\ value for $\mathcal I$, to the optimal value of \minPC, denoted $\opt^{\tiny\PC}$, i.e., the minimum number of cliques needed to cover the vertices of $G = (V, E)$. Let $V_1, \dots, V_{k^*}$ be a partition of $G$ into $k^*=\opt^{\tiny\PC}$ disjoint cliques.
Partition the vertices in $V'$ into 
coalitions $C_1', \dots, C'_{k^*}$ as evenly as possible; this ensures that 
$|C'_h|\ge \left\lfloor\frac{\modulus{V}}{k^*}\right\rfloor$ for each $h\in [k^*]$. Now set 
$C_h=V_h\cup C'_h$ for all $h\in[k^*]$, 
and note that $\pi= \set{C_1, \dots, C_{k^*}}$ is a clique partition of $\agents$. Since $V_h\neq\emptyset$ for all $h\in[k^*]$, we obtain $\modulus{C_h} \geq \left\lfloor\frac{\modulus{V}}{k^*}\right\rfloor+1\ge \frac{\modulus{V}}{k^*}$ for each $h\in [k^*]$. 
As the utility of each agent in a clique of size $t$ is $t-1$, we obtain
\begin{align}\label{eq:optESW}
    \opt(\instance) \geq \frac{\modulus{V}}{k^*} -1 = \frac{\modulus{V}}{\opt^{\tiny\PC}}-1.
\end{align}

Consider now the partition $\pi$ computed by our hypothetical approximation algorithm $\mathcal A$ for \maxEg, and let $s=\ESW(\pi)$. By \Cref{lemma:cliquePartitionPositiveESW}, $\pi$ must be a clique partition; moreover, $\ESW(\pi)=s$ implies that each coalition in $\pi$ contains at least $s+1$ vertices.
Hence, $\pi$ consists of at most $ \frac{2\cdot\modulus{V}}{s+1}$ disjoint cliques of $\Gsf$. We can now convert $\pi$ into a partition of $G$ into $k$ cliques with $k\leq \frac{2\cdot\modulus{V}}{s+1}$, as the restriction of each coalition $C\in \pi$ to the nodes of $V$ is either empty or a clique of $G$. Since $\mathcal A$ computes an $n^{1-\epsilon}$ approximation of \maxEg, we have 
\begin{align}\label{ratio:upper}
\frac{\opt(\instance)}{s} \leq n^{1-\epsilon}. 
\end{align}
On the other hand, 
$k\leq \frac{2\cdot\modulus{V}}{s+1}$ implies $s\le \frac{2\cdot\modulus{V}}{k}-1$; combining this with~\eqref{eq:optESW}, we get
\begin{align}\label{ratio:lower}
    \frac{\opt(\instance)}{s} \geq \frac{\frac{\modulus{V}}{\opt^{\tiny\PC}}-1}{\frac{2\cdot\modulus{V}}{k}-1} 
    \ge \frac{\frac{\modulus{V}}{\opt^{\tiny\PC}}}{\frac{2\cdot\modulus{V}}{k}-1}  - 1
    > \frac{k}{2\cdot\opt^{\tiny\PC}}  - 1 \, ,
\end{align}
    where the last two inequalities hold since $1\le \frac{2\cdot\modulus{V}}{k}-1 < \frac{2\cdot\modulus{V}}{k}$.

    As $n=2\cdot\modulus{V}$, 
    by putting together~\eqref{ratio:lower} and~\eqref{ratio:upper}, 
    we obtain $\frac{k}{\opt^{\tiny\PC}} < 2\cdot(2\cdot\modulus{V})^{1-\epsilon} + 2$. For sufficiently large $V$, this quantity is smaller than $\modulus{V}^{1-\epsilon/2}$, a contradiction to the inapproximability of \minPC.
\end{proof}

\subsection{ Linear Approximation in Polynomial Time}
By \Cref{lemma:cliquePartitionPositiveESW}, to obtain a bounded approximation for a given $\EA$ instance, we need to determine whether the associated strong friendship graph $\Gsf$ admits a {\em non-trivial clique partition}, i.e., a clique partition in which each agent belongs to a clique of size at least $2$. 

\begin{observation}\label{lemma:nonTrivialPartition}
    If $\Gsf$ does not admit a non-trivial clique partition, then a partition into singletons achieves optimal $\ESW$. Otherwise, any non-trivial clique partition guarantees an $\ESW$ of at least $1$.
\end{observation}

Note that a non-trivial clique partition exists if and only if it is possible to partition the graph into triangles and a matching. This observation enables us to design a
polynomial-time algorithm that computes an $(n-1)$-approximation to $\maxEg$.

\begin{theorem}
    Under $\EA$, there exists a polynomial-time algorithm computing an $(n-1)$-approximation to $\maxEg$. Furthermore, if $\ESW>0$, the computed partition satisfies \oneF\ for all the agents.
\end{theorem}
\begin{proof}
The problem of finding a perfect covering of the vertices of a graph by $\mathbb{K}_2$ and $\mathbb{K}_3$, i.e.\ cliques of size $2$ or $3$, admits a polynomial-time algorithm~\cite{hell1984packings}. 
We can execute this algorithm on $\Gsf$.
If the algorithm reports that there is no perfect covering of $\Gsf$ with $\mathbb{K}_2$ and $\mathbb{K}_3$, every partition of $\Gsf$ into cliques contains a singleton, and hence the maximum attainable egalitarian welfare is~$0$ (so, in particular, a partition into singletons maximizes $\ESW$).
On the other hand, if the algorithm returns a perfect covering, it corresponds to a partition of the agents into coalitions of size $2$ and $3$, which satisfies $\oneF$. Any such partition guarantees utility of at least $1$ to every agent. Since the maximum utility an agent can derive is at most $n-1$, the $\ESW$ of the partition constructed by our algorithm is within a factor of $n-1$ from optimal. 
\end{proof}

We also note that \maxEg\ becomes easy if $\Gsf$ is {\em triangle-free}, that is, no cycle in $\Gsf$ has length $3$. 

\begin{theorem}
   Under $\EA$, if $\Gsf$ is triangle-free, then \maxEg\ is in \classP.   
\end{theorem}
\begin{proof}
    In a triangle-free graph, no clique has size larger than~$2$. Thus, we proceed by computing a maximum matching in $\Gsf$. If $\Gsf$ admits a perfect matching, the maximum $\ESW$ is $1$, and an optimal partition is attained by creating a coalition for each pair of matched agents. Otherwise, the maximum $\ESW$ is $0$, and splitting the agents into singletons is an optimal partition. 
\end{proof}
\section{Friends Appreciation}
For $\FA$ preferences, unless specified otherwise, we will work with the $\Gf$ graph.
Let $f_{\min} = \min_{i \in \agents} f_i$ denote the minimum number of friends of any agent, 
and set 
$\agents_{\min} = \{i\in\agents\mid f_i=f_{\min}\}$.
We start by presenting simple bounds on $\opt$ for $\maxEg$;
for the upper bound, note that agents in $\agents_{\min}$ have at most $f_{\min}$ friends, and for the lower bound, consider the grand coalition.
\begin{restatable}{observation}{boundsOnOptFA}\label{obs:boundsOnOptFA}
Under $\FA$, 
$ f_{\min} - 1<\opt\leq  f_{\min}$. 
\end{restatable}
Consequently, in every optimal partition, every agent $i\in\agents_{\min}$ must end up in a coalition that contains all of her friends. 

\begin{corollary}\label{cor:allFfMin}
Under $\FA$, any \maxEg\ partition satisfies \allF\ for all the agents in $\agents_{\min}$.
\end{corollary}

\subsection{NP-Hardness}
 It turns out that for $\FA$ preferences, too, computing a \maxEg\ partition is \classNP-hard; the hardness result persists even if agents' preferences are symmetric.

\begin{theorem}\label{thm:OPThard}
   Under $\FA$, \maxEg\ is \classNP-hard. The hardness result holds even for symmetric instances.
\end{theorem}

Our proof proceeds by a reduction from the {\sc Partition into Triangles} problem, which can be formulated as follows:
{
\smallskip
\begin{center}
\fbox{
\begin{varwidth}{0.95\linewidth}
\noindent {{\sc Partition into Triangles}:}

\smallskip
\noindent\textit{Input:} An undirected graph $G=(V, E)$.\\
\smallskip
\noindent\textit{Question:} 
Is it possible to partition $V$
as $S_1, \dots, S_k$ so that for each $S_j$, $j\in[k]$, its induced subgraph is a triangle?
\end{varwidth}
}
\end{center}
\smallskip
}
{\sc Partition into Triangles} is known to be NP-hard~\cite{garey1979computers}.
Note that we can assume that $|V|=3k$ for some $k\in\mathbb N$, as otherwise we trivially have a ``no''-instance. Thus, in what follows, we will make this assumption.

\begin{proof}
    Consider an instance $G=(V, E)$ of {\sc Partition into Triangles} with $|V|=3k$. We will construct a (symmetric) instance of $\FA$ by specifying its associated graph $\Gsf$.

    For each vertex $v\in V$ we construct four vertices of $\Gsf$: $a_v$, $b_v$, $c_v$ and $z_v$. Thus, $\agents=\{a_v, b_v, c_v, z_v\mid v\in V\}$ and $n=\modulus{\agents}=4\modulus{V}$.
    For each $v\in V$, the graph $\Gsf$ contains edges $\{a_v, z_v\}$, $\{b_v, z_v\}$, and $\{c_v, z_v\}$. Also, for each edge $\{u, v\}\in E$ of $G$, the graph $\Gsf$ contains edges  $\{a_u, a_v\}$, $\{b_u, b_v\}$, and $\{c_u, c_v\}$. This completes the description of $\Gsf$.

    We claim that the instance of $\FA$ that corresponds to $\Gsf$ admits a partition $\pi$ with $\ESW(\pi)\ge 3-\frac{8}{n}$ if and only if $G$ is a yes-instance of {\sc Partition into Triangles}.

    Indeed, let $S_1, \dots, S_k$ be a partition of $G$ into triangles. For each $S_j$, $j\in[k]$, we create a coalition $C_j=\{a_v, b_v, c_v, z_v\mid v\in S_j\}$. This coalition contains $12$ agents, and each agent in it has $3$ friends and 8 enemies, so the utility of each agent is $3-\frac{8}{n}$.

    Conversely, suppose our $\FA$ game admits a partition $\pi$ with $\ESW(\pi)\ge 3-\frac{8}{n}$. Since $n=4|V|\ge 12$, we have $3-\frac{8}{n}>2$, so in this partition, each agent must be in a coalition with at least $3$ friends. Thus, for each $v\in V$ it holds that $z_v$ is in the same coalition as $a_v$, $b_v$ and $c_v$ (as $z_v$ has no other friends). It follows that each coalition in $\pi$ is of the form $C = \{a_v, b_v, c_v, z_v\mid v\in S\}$, where $S$ is a subset of $V$. Moreover, $z_v$'s coalition must contain at most $8$ enemies, implying $\modulus{S}\leq 3$. Now, if $\modulus{S}<3$ or if its induced graph is not a triangle, there is an agent $a_v$ in $C$ that has fewer than $3$ friends in $C$, and hence her utility is lower than $3-\frac{8}{n}$, a contradiction. Thus, each coalition in $\pi$ corresponds to a triangle in $G$, and these triangles are pairwise disjoint, i.e., $G$ admits a partition into triangles.
\end{proof}

In the remainder of this section, we present our approximability results.
We describe a simple algorithm whose approximation ratio is linear in the number of agents. While the linear approximation bound is tight for our algorithm, this is caused by instances with $f_{\min}=1$: we show that as soon as each agent has at least two friends, the approximation ratio of our algorithms becomes $2-\Theta\left(\frac{1}{n}\right)$. We also obtain $\left(2-\Theta\left(\frac{1}{n}\right)\right)$-approximation results for two other settings: for a randomized algorithm (in expectation) and for symmetric instances.

\subsection{Deterministic Bounded Approximation}
\Cref{cor:allFfMin} establishes that, in an optimal partition, the agents with the minimum number of friends, i.e., the ones in $\agents_{\min}$, must be in a coalition with all their friends.
As the marginal contribution of all enemies is, in absolute terms, less than that of a single friend, a naive approach would be to partition the agents so that each agent---not only the ones in $\agents_{\min}$---is placed in a coalition together with all of their friends, i.e., to create a partition that satisfies $\allF$. Among all such partitions, we prefer ones that avoid placing enemies in the same coalition as much as possible.
Thus, an attractive strategy is to create a coalition for each weakly connected component of $\Gf$; we refer to the resulting algorithm as \weaklyC:

\paragraph{\weaklyC:} Given a graph $\Gf$, let $\pi=(C_1, \dots, C_m)$ be the list of weakly connected components of $\Gf$; return the partition $\pi$. 

\smallskip 

Unfortunately, we will show that this algorithm does not always provide a constant-factor approximation to $\maxEg$.

\begin{theorem}\label{thm:apxWC}
Under $\FA$, the algorithm {\em \weaklyC} guarantees an approximation of $2 - \frac{6}{n+3}$ to $\maxEg$ if $f_{\min} \geq 2$ and $\frac{n}{2}$ otherwise.
\end{theorem}
\begin{proof}
    Consider an instance $\mathcal I$ of $\FA$. By \Cref{rem:Fneqempty}, we can assume that $F_i \neq \emptyset$ for all $i \in \agents$. Let $\pi$ denote the outcome returned by \weaklyC\ on $\mathcal I$.

     By \Cref{obs:boundsOnOptFA}, we have $\opt \leq f_{\min}$.
     On the other hand, in $\pi$, every agent $i\in\agents$ is in a coalition with $f_i$ friends and at most $n- f_i -1$ enemies. Therefore, 
     \[
     u_i(\pi)\geq f_i - \frac{n-f_i-1}{n} = f_i\cdot\frac{n+1}{n} -\frac{n-1}{n} \, .
     \]
    Let $j$ be an agent with minimum utility in $\pi$; note that $j\in\agents_{\min}$. The approximation ratio of $\weaklyC$ on $\mathcal I$ can be bounded as
    \begin{align*}
        \frac{\opt(\instance)}{\min_{i\in\agents} u_i(\pi)} \leq  \frac{f_j}{u_j(\pi)} \leq \frac{f_j}{f_j\cdot\frac{n+1}{n} -\frac{n-1}{n}} =\frac{f_j\cdot n}{f_j (n+1)-(n-1)}.
    \end{align*}
    The ratio $\frac{f_j \cdot n}{f_j(n+1) -(n-1)}$ is a decreasing function of $f_j$, so it takes its maximum value, which is $\frac{n}{2}$, if $f_j=1$. If $f_{\min}\ge 2$ (and hence $f_j\ge 2$), 
    we have
    $$
    \frac{f_j\cdot n}{f_j(n+1) - (n-1)} \le  \frac{2n}{2(n+1)-(n-1)}=\frac{2n}{n+3}=2-\frac{6}{n+3}.
    $$ 
This completes our proof.  
\end{proof}
In the appendix, we give an example showing that both bounds in the statement of \Cref{thm:apxWC} are tight.

\subsection{Assigning Agents with Only One Friend}
\Cref{thm:apxWC} establishes that the most challenging instances to approximate are those containing an agent who only has one friend. Indeed, a constant approximation can be achieved as soon as $\modulus{F_i}\geq 2$ for all $i\in\agents$.
In the next sections, we show that constant-factor approximation can be recovered in randomized and symmetric settings.
To this end, we need to tackle the problem of how to assign agents who only have one friend.  Hereafter, we 
denote by $\agents_1(\instance) =\set{i\in\agents \mid  f_i  = 1}$ the set of agents with a unique friend. When $\instance$ is clear from the context, we simply write $\agents_1$.

If $\agents_1\neq \emptyset$, \Cref{obs:boundsOnOptFA} and \Cref{cor:allFfMin} establish that, to obtain any approximation, we must satisfy \oneF$\equiv$\allF\ for all agents in $\agents_1$ and, moreover, $\opt \in (0,1]$. Given an $\FA$ instance $\instance$, we can guarantee \oneF\ to all agents in $\agents_1$ by placing all agents in the same coalition. A more sophisticated, yet tractable, approach is to create a minimal-by-refinement partition that enforces \oneF\ for all agents in~$\agents_1$; we will refer to this algorithm as \oneFalgo.

\paragraph{\oneFalgo:} 
On an input $\Gf=(\agents, F)$, the algorithm creates the graph $G_1^f=(\agents, F_1)$ by only maintaining the friend relationships of the agents in $\agents_1$, i.e., it sets $F_1=\set{(i,j)\in F \mid i\in \agents_1 }$.  Then it splits the agents into the weakly-connected components of $G_1^f,$ that is, applies \weaklyC\ on $G_1^f$. 

\smallskip 

By construction, the partition computed by \oneFalgo\ is the unique minimal-by-refinement partition that guarantees \oneF\ to all agents in $\agents_1$.
Indeed, if a partition $\pi$ satisfies \oneF\ for all agents in $\agents_1$, then every coalition $S\in \pi$ must contain a weakly connected component of $G_1^f$. Moreover, if $S$ contains several weakly connected components of $G_1^f$, then $\pi$ is not a minimal-by-refinement partition that satisfies \oneF, as $S$ can be split into smaller coalitions---one for each weakly connected component of $G_1^f$ that it contains.

Next, we establish a useful connection between the utility of agents in $\agents_1$ in 
the minimal-by-refinement \oneF\ partition and in 
an optimal partition.

\begin{lemma}\label{lemma:LBoneF}
    Let $\pi$ be the unique minimal-by-refinement partition that provides \oneF\ to all agents in $\agents_1$. Then, for all $i\in\agents_1$ we have
    $u_i(\pi)\geq \opt$, 
    and $\pi(i)\subseteq \pi^*(i)$ for every optimal partition $\pi^*$.
\end{lemma}
\begin{proof}
Assume $\agents_1 \neq \emptyset$; otherwise, the statement is immediate.

Let $\pi^*$ be an optimal partition.
By \Cref{obs:boundsOnOptFA}  we have $\opt>0$; hence, $\pi^*$ satisfies condition \oneF\ for all agents.
As previously observed, for \oneF\ to hold for all agents in $\agents_1$, the members of each weakly connected component of $G_1^f$ must belong to the same coalition.
Moreover, we have already established that the unique minimal-by-refinement partition $\pi$ satisfying \oneF\ for the agents in $\agents_1$ is precisely the partition where each coalition is a weakly connected component of $G_1^f$.
Consequently, for every $i \in \agents_1$ we have $\pi(i) \subseteq \pi^*(i)$.
Moreover, 
since $F_i \subseteq \pi(i) \subseteq \pi^*(i)$, we obtain $u_i(\pi) \geq u_i(\pi^*) \geq \opt$.
\end{proof}

By \Cref{lemma:LBoneF}, $\oneFalgo(\instance)$ guarantees utility of at least $\opt$ to the agents in $\agents_1$; however, this algorithm may place some agents in $\agents \setminus\agents_1$ in singleton coalitions, so that their utility is $0$.

Clearly, if $\agents_1=\agents$, then, by \Cref{lemma:LBoneF}, \oneFalgo\ computes a \maxEg\ partition. This leads to the following theorem.

\begin{theorem}
    Under $\FA$, if $f_i=1$ for all $i\in\agents$ then \maxEg\ is in \classP. Furthermore, $\OPT$ is the unique minimal-by-refinement partition satisfying \oneF\ (which in this case coincides with \allF) for all agents.
\end{theorem}

\subsection{Randomized Constant Approximation}
In what follows, we combine \oneFalgo\ and \weaklyC\ to obtain a randomized algorithm
whose approximation ratio is at most $2$.
To start, we need to define the approximation ratio in the randomized setting. 
Since for a given instance $\instance$ the space of all possible outcomes is finite, a randomized algorithm $\mathcal{A}$ can be viewed as a probability distribution over a finite set of deterministic algorithms $\Sigma$, 
where each algorithm in $\Sigma$ computes a partition of the agents in~$\instance$;
we write $\mathcal{A} \sim \Delta(\Sigma)$. 

There are two natural approaches to measuring the performance of a randomized algorithm in our setting: 
(1) the expected minimum utility, over the random choices of the algorithm, as given by 
$\exptDistr{\min_{i\in\agents} u_i(\mathcal{A}(\instance))}{\mathcal{A} \sim \Delta(\Sigma)}$, and
(2) the minimum expected utility over all agents, as given by $\min_{i\in\agents}\exptDistr{ u_i(\mathcal{A}(\instance))}{\mathcal{A} \sim \Delta(\Sigma)}$.

It is easy to see that the first measure is more stringent, i.e., 
$$
\exptDistr{\min_{i\in\agents} u_i(\mathcal{A}(\instance))}{\mathcal{A} \sim \Delta(\Sigma)}\leq 
\min_{i\in\agents} \exptDistr{ u_i(\mathcal{A}(\instance))}{\mathcal{A} \sim \Delta(\Sigma)}.
$$

Indeed, under the first approach we cannot expect to get an improved approximation ratio from randomization, since 
for every instance $\instance$ we have
\begin{align*}
    \exptDistr{\min_{i\in\agents} u_i(\mathcal{A}(\instance))}{\mathcal{A} \sim \Delta(\Sigma)} 
    \leq \max_{\mathcal{A}\in\Sigma}\min_{i\in\agents} u_i(\mathcal{A}(\instance)) \, ,
\end{align*}
i.e., the expected egalitarian welfare of a randomized algorithm $\mathcal{A} \sim \Delta(\Sigma)$ cannot be better than the egalitarian welfare achieved deterministically by the best algorithm in $\Sigma$.

Thus, we pursue the second approach, and define the approximation ratio as follows:
\begin{align*}
    r_n(\mathcal{A} \sim \Delta(\Sigma)) = \sup_{\substack{\instance=(\agents, F):\\|\agents|=n}} \frac{\opt(\instance)}{\min_{i\in\agents} \exptDistr{ u_i(\mathcal{A}(\instance))}{\mathcal{A} \sim \Delta(\Sigma)}} \, , 
\end{align*} 
and $r(\mathcal{A} \sim \Delta(\Sigma))=\sup_{n\in\mathbb N}r_n(\mathcal{A} \sim \Delta(\Sigma))$.

We are ready to define our randomized algorithm,
\randAlgo.
\paragraph{\randAlgo:}
With probability $\alpha$, executes algorithm $\weaklyC$, and with probability $1-\alpha$ it executes algorithm $\oneFalgo$, 
where the parameter $\alpha \in [0,1]$ will be set later and will be dependent on the input instance.

\begin{theorem}\label{thm:approxRandom}
    Under $\FA$, 
    there exists  
    a choice of 
    $\alpha$ such that 
    \[
    r_n(\randAlgo) \leq 2- \frac{5}{n+3} \, .
    \]
\end{theorem}
\begin{proof}
 Given an instance $\instance$, let us set $\pi=\weaklyC(\instance)$ and $ \pi'=\oneFalgo(\instance)$. Moreover, let $v=\min_{i\in \agents_1 } u_i(\pi')$.

 If $\agents_1= \emptyset$, we set $\alpha=1$: then \randAlgo\ coincides with \weaklyC, so its approximation ratio for the case $\agents_1= \emptyset$ is bounded by $2- \frac{6}{n+3}$.

 Assume now $\agents_1\neq \emptyset$. We will estimate the expected utility of an agent $i\in\agents$ depending on whether $i\in \agents_1$ or $i\in \agents\setminus \agents_1$.

 If $i\in \agents_1$, then $u_i(\pi)\geq \frac{2}{n}$, as shown in the proof of \Cref{thm:apxWC}, while $u_i(\pi')\geq v$. Hence, $\expectation{u_i(\randAlgo)}\geq \alpha\cdot \frac{2}{n} + (1-\alpha)\cdot v$.

 If $i\in \agents\setminus \agents_1$, then $u_i(\pi)\geq 1 + \frac{3}{n}$, as in the worst case $i$ has only two friends and is put in a coalition with all her enemies, while $u_i(\pi')$ may be $0$. Consequently, $\expectation{u_i(\randAlgo)}\geq \alpha\cdot  \left( 1 +\frac{3}{n}\right)$.

 Putting these two scenarios together, we get 
 \begin{align*}
     \min_{i\in\agents}\expectation{u_i(\randAlgo(\instance))} &\geq 
     \min \set{ v + \alpha\cdot\left( \frac{2}{n} - v\right), \alpha\cdot  \left( 1 +\frac{3}{n}\right)} \, . 
 \end{align*}

Since $v \geq \tfrac{2}{n}$, the function $v + \alpha \cdot \left(\tfrac{2}{n} - v\right)$ is decreasing in $\alpha$, whereas $\alpha \cdot \left(1 + \tfrac{3}{n}\right)$ is increasing in $\alpha$. The minimum of the two is attained when they are equal, which happens for $\alpha= \frac{v}{1+\frac{1}{n} +v }$. Hence,
\begin{equation}\label{eq:rand}
\min_{i\in\agents}\expectation{u_i(\randAlgo(\instance))}\geq \left( 1 +\frac{3}{n}\right)\cdot \frac{v}{1+\frac{1}{n} +v } \, .
\end{equation}
The right-hand side of~\eqref{eq:rand} is an increasing function of $v$ and by \Cref{lemma:LBoneF} we have $v\geq \opt(\instance)$. Hence,  $$ \min_{i\in\agents}\expectation{u_i(\randAlgo(\instance))}\geq \left( 1 +\frac{3}{n}\right)\cdot \frac{\opt(\instance)}{1+\frac{1}{n} +\opt(\instance) }.
$$ 
Now, $\agents_1\neq \emptyset$ implies $\opt(\instance) \leq 1$, so we obtain
\begin{align*}
    r_n(\randAlgo) &= \sup_{\substack{\instance=(\agents, F)\\|\agents|=n}} \frac{\opt(\instance)}{\min_{i\in\agents} \expectation{ u_i(\randAlgo(\instance))}} \\
    &\leq \frac{1+\frac{1}{n} +\opt(\instance)}{1 +\frac{3}{n}} \leq\frac{2n +1}{n+3}
    =2- \frac{5}{n+3}\, 
\end{align*}
which is what we wanted to prove.
\end{proof}

In the appendix, we give an example showing the tightness of the analysis.

\subsection{Deterministic Constant Approximation in the Symmetric Case}
In this section, we will be working with symmetric instances; therefore, we will only consider the graph $\Gsf$.

\begin{theorem}\label{thm:symmetricFA}
Under $\FA$, there exists a polynomial-time algorithm that computes a $(2-\frac{4}{n+2})$-approximation for symmetric instances.
\end{theorem}
\begin{proof}
If $\agents_1= \emptyset$, then the approximation ratio of \weaklyC\ (which in this case simply computes the connected components of $\Gsf$) is $2-\frac{6}{n+3} < 2-\frac{4}{n+2}$.
Thus, hereafter, we assume $\agents_1\neq \emptyset$.  
Our algorithm proceeds in three stages.

{\bf Stage 1:\ } Let ${\mathcal R} = \set{j\mid j\in F_i\text{ for some }i\in\agents_1}$ and $\overline{\agents}_1=\agents_1\cup\mathcal R$; that is, $\overline{\agents}_1$ contains the agents having only one friend as well as their unique friends. 
In the first stage, we construct the minimal partition of the agents in $\overline{\agents}_1$ that places each agent in $\agents_1$ together with her unique friend.
To this end, for each $r\in \mathcal R$ let $S(r)=\{i\in\agents_1\mid r\in F_i\}$, and note that the sets $\{S(r)\}_{r\in {\mathcal R}}$ form a partition of 
$\overline{\agents}_1$, which we will denote by $\pi_1$ (if two agents $i_1, i_2\in \agents_1$ are each other's unique friends, they both appear in $\mathcal R$; the sets $S(i_1)$ and $S(i_2)$ then coincide, and $\pi_1$ includes a single copy of this set).
By construction, each $S(r)$, $r\in \mathcal R$, induces a star in $\Gsf$.
Furthermore, the utility of the root of a star is lower-bounded by the utility of its leaves, all of whom share the same utility. Note that, by \Cref{lemma:LBoneF}, for each partition $\pi^*$ of $\agents$ that maximizes $\ESW$ and for each $i\in\overline{\agents}_1$ we have 
$\pi_1(i)\subseteq \pi^*(i)$ and hence $u_i(\pi_1)\ge\opt$.

{\bf Stage 2:\ } Next, let ${\mathcal L}=\{i\in \agents\setminus\overline{\agents}_1\mid F_i\subseteq \overline{\agents}_1\}$ be the set of {\em lonely} agents in $\agents\setminus\overline{\agents}_1$, i.e., agents with no friends in $\agents\setminus\overline{\agents}_1$. In the second stage,
we deal with agents in $\mathcal L$. Note that, for each $i\in \mathcal L$ we have $F_i\subseteq \mathcal R\setminus \agents_1$, so, 
to obtain positive utility, agents in $\mathcal L$ need to be added to coalitions in $\pi_1$. We add agents in $\mathcal L$ to coalitions in $\pi_1$ so as to minimize the size of the largest coalition in the resulting partition of $\overline{\agents}_1\cup \mathcal L$, subject to the condition 
that each agent in $\mathcal L$ is placed in a coalition with one of her friends.

To this end, we construct a bipartite graph with parts $\mathcal L$ and $\mathcal R$ (where there is an edge between $i\in \mathcal L$ and $r\in \mathcal R$ if and only if $\{i, r\}\in\Gsf$), and solve a sequence of b-matching problems on this graph. Specifically, we iterate through $\kappa=\max_{r\in R}|S_r|, \dots, n$ and check if this graph admits a b-matching of size $|{\mathcal L}|$ where the capacity of each $i\in \mathcal L$ is $1$ and the capacity of $r\in \mathcal R$ is $\kappa-|S_r|$; we identify the smallest $\kappa$ for which such a b-matching exists and place each $i\in \mathcal L$ into the coalition $S(r)$ such that $i$ is matched to $r$. This step can be executed in polynomial time, since the problem of finding a b-matching in a bipartite graph is polynomial-time solvable. 

Let $\pi_2$ be the resulting partition of $\overline{\agents}_1\cup\mathcal L$; by construction, for each agent $i\in \overline{\agents}_1\cup \mathcal L$ we have $u_i(\pi_2)\ge \opt$.

{\bf Stage 3:\ } Finally, let ${\mathcal U}=\agents\setminus({\mathcal L}\cup \overline{\agents}_1)$ be the set of  agents who remained unassigned by the end of Stage~2. In the third stage, we deal with agents in $\mathcal U$. Note that, by construction of $\mathcal L$ (and by symmetry), no agent in $\mathcal U$ has a friend in $\mathcal L$, and hence each agent in $\mathcal U$ has at least one friend in $\mathcal U$. 

If every agent in $\mathcal U$ has at least $2$ friends in $\mathcal U$, or if $\modulus{{\mathcal U}} \leq \frac{n}{2} +1$, we put all agents in $\mathcal U$ in the same coalition ${\mathcal U}_0:= \mathcal U$, add this coalition to $\pi_2$, and return the resulting partition $\pi_3=\pi_2\cup\{{\mathcal U}_0\}$.

Otherwise, we repeat 
the step ($*$) described below until either $\modulus{{\mathcal U}} \le \frac{n}{2} +1$ or all agents in $\mathcal U$ have at least two friends in $\mathcal  U$:

\begin{itemize}
\item[($*$)] Pick an agent $i\in \mathcal U$ with at most one friend in $\mathcal U$. If the unique friend of $i$ in $\mathcal U$ is $j$ and $F_j\cap {\mathcal U}=\{i\}$, add $\{i, j\}$ to the partition and remove $i, j$ from $\mathcal U$. Otherwise, identify
a coalition $S$ in the current partition that contains a friend of $i$ (this is always possible since by construction of $\mathcal U$ we have $|F_i|\ge 2$ for each $i\in \mathcal U$), add $i$ to $S$, and remove $i$ from $\mathcal U$. 
\end{itemize}

Observe that after step ($*$) our partition may contain coalitions that are not stars: e.g., the first time we execute this step, we may add an agent $i$ to a star $S(r)$ such that $r$ is a friend of $i$, but at a later iteration of step ($*$) we may add a friend of $i$ to this coalition.

Let ${\mathcal U}_0$ be what remains of $\mathcal U$ at the end of this process. Place all agents in ${\mathcal U}_0$ in the same coalition, and add it to the partition; let $\pi_3$ be  
 the resulting coalition structure.
This completes the description of our algorithm. We now claim that its approximation ratio is at most $2-\frac{4}{n+2}$. 

Note that in $\pi_3$ each agent is in a coalition with at least one friend, and consider a coalition $S\in \pi_3$.

If $|S|\leq \frac{n}{2} +1$, 
then for each $i\in S$ we have 
\begin{align}\label{eq:symmetricApx}
u_i(S)\geq 1 - \frac{\modulus{S}-2}{n}\geq 1- \frac{\frac{n}{2} -1}{n} = \frac{1}{2} + \frac{1}{n} = \frac{n+2}{2n}\, .    
\end{align}

On the other hand, suppose that $|S|> \frac{n}{2} +1$. This can happen if (i) $S$ was constructed during stages 1 and 2, 
or (ii) $S={\mathcal U}_0$ and each agent in ${\mathcal U}_0$ has at least two friends in ${\mathcal U}_0$. Note, in particular, that we cannot create coalitions of that size in step ($*$), since this step is only executed when there are at least $\frac{n}{2}+1$ agents in $\mathcal U$.
In case (i), we have argued that $u_i(S)\ge\opt$ for each $i\in S$, and in case (ii), the utility of each agent in $S$ is at least $2-\frac{|S|-3}{n}>1\ge\opt$.

Thus, if for some agent $i\in\agents$ we have $u_i(\pi_3)<\opt$, then 
$u_i(\pi_3)\ge \frac{n+2}{2n}$. On the other hand, $\agents_1\neq\emptyset$ implies $\opt\le 1$. Hence,   
the approximation ratio of our algorithm does not exceed $\frac{2n}{n+2}=2-\frac{4}{n+2}$, as desired.
\end{proof}

For completeness, the appendix provides the pseudocode of the algorithm described in the proof of \Cref{thm:symmetricFA} together with a discussion of its polynomial-time complexity. Moreover, we also demonstrate that our analysis is tight. 

Finally, we show that, for symmetric instances where $\Gsf$ is a forest, the problem of finding a \maxEg\ partition is polynomial-time solvable. To this end, we establish a connection to a different class of hedonic games, namely, {\em fractional hedonic games (FHGs)}~\cite{aziz2019fractional}. In these games, too, agents assign values to each other, with agent $i$ assigning value $v_i(j)\in\mathbb R$ to agent $j$, and the utility an agent $i\in \agents$ assigns to a coalition $C\in\agents^i$ is computed as $\frac{1}{|C|}\sum_{j\in C\setminus\{i\}}v_i(j)$. An FHG with a set of agents $\agents$ is {\em simple} if $v_i(j)\in\{0, 1\}$ for all $i,j\in\agents$; it is {\em symmetric} if $v_i(j)=v_j(i)$ for all $i, j\in\agents$. A simple symmetric FHG $\mathcal G$ can be represented by an undirected graph $G({\mathcal G})=(\agents, {\mathcal E})$ where $\{i, j\}\in\mathcal E$ iff $v_i(j)=v_j(i)=1$.
Thus, an undirected graph $G$ with vertex set $\agents$ induces two different games on the set of agents $\agents$: a symmetric FA game and a simple symmetric FHG. 
We will now establish a tight connection between the structure of optimal solutions of these games when $G$ is acyclic.

\begin{lemma}\label{lem:fhg}
Consider a simple symmetric FHG ${\mathcal G}_1$ and a symmetric FA ${\mathcal G}_2$ 
with the same set of agents $\agents$ such that $G({\mathcal G}_1)$ is acyclic and
coincides with the strong friendship graph of ${\mathcal G}_2$. Then any partition $\pi^*$ that maximizes the $\ESW$ in ${\mathcal G}_1$ can be transformed in polynomial time into a partition $\pi'$ that 
maximizes the $\ESW$ in ${\mathcal G}_2$.
\end{lemma}
\begin{proof}
Assume without loss of generality that $G({\mathcal G}_1)$ is a tree; otherwise, we apply the subsequent argument to each connected component of $G({\mathcal G}_1)$.   

Let $\pi$ be a partition of $\agents$, and consider a coalition $C\in\pi$.
Suppose first that $C$ is not connected. Then splitting $C$ into its connected components is a Pareto improvement for all agents in both ${\mathcal G}_1$ and ${\mathcal G}_2$ and hence does not decrease the $\ESW$ in either game. Now, suppose $C$ is not a star. Then $C$ has two non-leaf agents connected by an edge; let $e$ be some such edge, and let $C'$ and $C''$ be the coalitions obtained by splitting $C$ along $e$. In each of $C$, $C'$ and $C''$, the lowest-utility vertices are leaves (this holds both in FA games and in FHGs), and the utility of a leaf in $C'$ and $C''$ is higher than the utility of a leaf in $C$.
Thus, this split, too, does not decrease the $\ESW$ in either game. It follows that we can split each coalition in $\pi$ into stars without lowering the $\ESW$ in either game. 

Now, let $\pi^*$ be a partition that maximizes $\ESW$ in ${\mathcal G}_1$.
Let $\pi'$ be a partition obtained from $\pi^*$ by splitting each non-star coalition in $\pi^*$ into stars of size at least $2$; note that $\pi'$
can be obtained from $\pi^*$ in polynomial time. Moreover, the argument above shows that $\pi'$ maximizes $\ESW$ in ${\mathcal G}_1$.
Observe that, in both games, the agents with the lowest utility in $\pi'$ are the leaves of the largest stars in $\pi'$.

We claim that $\pi'$ is optimal for ${\mathcal G}_2$. Indeed, suppose not, and let $s$ be the size of the largest star in $\pi'$. Let $\pi''$ be an optimal partition for ${\mathcal G}_2$. By the argument above,
we can assume that $\pi''$ is a collection of stars, so for it to have a higher $\ESW$ in ${\mathcal G}_2$ than $\pi'$, the largest star in $\pi''$ must be of size $s''<s'$. But then the $\ESW$ of $\pi''$ in ${\mathcal G}_1$ is higher than that of $\pi'$, a contradiction.
\end{proof}

If the graph $G({\mathcal G})$ associated with an FHG $\mathcal G$ is acyclic, a \maxEg\ partition can be computed in polynomial time~\cite{aziz2015welfare,hanaka2019computational}. Together with \Cref{lem:fhg}, this implies the following result.
\begin{theorem}
    Under $\FA$, if the game instance is symmetric and $\Gsf$ is a forest, then \maxEg\ is in \classP.
\end{theorem}

While we have just shown that, when the underlying graph is a tree, solving \maxEg\ is equivalent in $\FA$ and in sFHGs, this equivalence no longer holds for general graph structures.

\begin{example}
    Consider a symmetric instance where $n\geq 4$ is an even number and $\Gsf$ is an $n$-vertex cycle. Under $\FA$, the unique optimum is
    the grand coalition, whereas in sFHGs, every optimal partition corresponds to a perfect matching in $\Gsf$.
\end{example}
\section{Conclusions and Future Work}
We studied the problem of maximizing the egalitarian welfare in Friends and Enemies Games, focusing on two classical preference models: Enemies Aversion ($\EA$) and Friends Appreciation ($\FA$).
We provided an almost complete picture of the complexity of this problem, by establishing both hardness results and polynomial-time approximation guarantees.
A natural next step is to determine whether a constant-factor approximation is indeed possible for $\FA$ or to derive an inapproximability result. Moreover, employing the egalitarian welfare as a measure of the quality of stable outcomes, similarly to the work of Monaco et al.~\cite{monaco2019performance}, also represents a promising direction for future work.
More broadly, it would be interesting to explore egalitarian welfare in other variants of Friends and Enemies Games, such as those that incorporate neutral agents (with small positive or negative effects) or assign a different negative value for enemies instead of the standard $-\frac{1}{n}$ or $-n$.
Finally, while the utilitarian and egalitarian objectives have received attention, the Nash welfare in hedonic games remains unexplored.

\section*{acknowledgments}
This work was partially supported by:  
PNRR MIUR project FAIR - Future AI Research (PE00000013), Spoke 9 - Green-aware AI; PNRR MIUR project VITALITY (ECS00000041), Spoke 2 - Advanced Space Technologies and Research Alliance (ASTRA);
the European Union - Next Generation EU under the Italian PNRR, Mission 4, Component 2, Investment 1.3, CUP J33C22002880001, partnership on ``Telecommunications of the Future''(PE00000001 - program ``RESTART''), project MoVeOver/SCHEDULE (``Smart interseCtions witH connEcteD and aUtonomous vehicLEs'', CUP J33C22002880001);
and GNCS-INdAM project, CUP\_E53C24001950001.

\bibliographystyle{ACM-Reference-Format} 
\bibliography{references}

@article{dreze1980hedonic,
  title={Hedonic coalitions: Optimality and stability},
  author={Dreze, Jacques H and Greenberg, Joseph},
  journal={Econometrica: Journal of the Econometric Society},
  pages={987--1003},
  year={1980},
  publisher={JSTOR}
}

@article{bogomolnaia2002stability,
  title={The stability of hedonic coalition structures},
  author={Bogomolnaia, Anna and Jackson, Matthew O},
  journal={Games and Economic Behavior},
  volume={38},
  number={2},
  pages={201--230},
  year={2002},
  publisher={Elsevier}
}

@article{zuckerman2007linear,
  author       = {David Zuckerman},
  title        = {Linear degree extractors and the inapproximability of max clique and
                  chromatic number},
  journal      = {Theory of Computing},
  volume       = {3},
  number       = {1},
  pages        = {103--128},
  year         = {2007}
}

@inproceedings{elkind2009hedonic,
  title={Hedonic coalition nets},
  author={Elkind, Edith and Wooldridge, Michael},
  booktitle={Proceedings of The 8th International Conference on Autonomous Agents and Multiagent Systems, AAMAS 2009},
  pages={417--424},
  year={2009}
}

@article{elkind2020price,
  title={Price of Pareto optimality in hedonic games},
  author={Elkind, Edith and Fanelli, Angelo and Flammini, Michele},
  journal={Artificial Intelligence},
  volume={288},
  pages={103357},
  year={2020},
  publisher={Elsevier}
}

@inproceedings{gairing2010computing,
  title={Computing stable outcomes in hedonic games},
  author={Gairing, Martin and Savani, Rahul},
  booktitle={Proceedings of the 3rd International Symposium on Algorithmic Game Theory, SAGT 2010},
  pages={174--185},
  year={2010}
}

@article{banerjee2001core,
  title={Core in a simple coalition formation game},
  author={Banerjee, Suryapratim and Konishi, Hideo and S{\"o}nmez, Tayfun},
  journal={Social Choice and Welfare},
  volume={18},
  number={1},
  pages={135--153},
  year={2001},
  publisher={Springer}
}

@incollection{AzizS16,
	author    = {Haris Aziz and
	Rahul Savani},
	title     = {Hedonic games},
	booktitle = {Handbook of Computational Social Choice},
	pages     = {356--376},
	year      = {2016}
}

@article{hell1984packings,
  title={Packings by cliques and by finite families of graphs},
  author={Hell, Pavol and Kirkpatrick, David G},
  journal={{Discrete Mathematics}},
  volume={49},
  number={1},
  pages={45--59},
  year={1984},
  publisher={Elsevier Science Publishers BV Amsterdam, The Netherlands, The Netherlands}
}

@article{aziz2013computing,
  title={Computing desirable partitions in additively separable hedonic games},
  author={Aziz, Haris and Brandt, Felix and Seedig, Hans Georg},
  journal={Artificial Intelligence},
  volume={195},
  pages={316--334},
  year={2013},
  publisher={Elsevier}
}

@inproceedings{peters2016graphical,
  title={Graphical hedonic games of bounded treewidth},
  author={Peters, Dominik},
  booktitle={Proceedings of the 30th AAAI Conference on Artificial Intelligence, AAAI 2016},
  pages        = {586--593},
  year={2016}
}

@inproceedings{monaco2019performance,
  title={On the performance of stable outcomes in modified fractional hedonic games with egalitarian social welfare},
  author={Monaco, Gianpiero and Moscardelli, Luca and Velaj, Yllka},
  booktitle={Proceedings of the 18th International Conference on Autonomous Agents and Multiagent Systems, AAMAS 2019},
  pages={873--881},
  year={2019}
}

@article{aziz2019fractional,
  title={Fractional hedonic games},
  author={Aziz, Haris and Brandl, Florian and Brandt, Felix and Harrenstein, Paul and Olsen, Martin and Peters, Dominik},
  journal={ACM Transactions on Economics and Computation (TEAC)},
  volume={7},
  number={2},
  pages={1--29},
  year={2019},
  publisher={ACM New York, NY, USA}
}

@inproceedings{peters2016complexity,
  title={Complexity of hedonic games with dichotomous preferences},
  author={Peters, Dominik},
  booktitle={Proceedings of the 30th AAAI Conference on Artificial Intelligence, AAAI 2016},
  pages        = {579--585},
  year={2016}
}

@article{aziz2013pareto,
  title={Pareto optimality in coalition formation},
  author={Aziz, Haris and Brandt, Felix and Harrenstein, Paul},
  journal={Games and Economic Behavior},
  volume={82},
  pages={562--581},
  year={2013},
  publisher={Elsevier}
}

@inproceedings{constantinescu2023solving,
  author       = {Andrei Constantinescu and
                  Pascal Lenzner and
                  Rebecca Reiffenh{\"{a}}user and
                  Daniel Schmand and
                  Giovanna Varricchio},
  title        = {Solving Woeginger's hiking problem: Wonderful partitions in anonymous
                  hedonic games},
  booktitle    = {Proceedings of the 51st International Colloquium on Automata, Languages, and Programming,
                  {ICALP} 2024},
  pages        = {48:1--48:18},
  year         = {2024}
}

@article{waxman2020maximizing,
  title={On maximizing egalitarian value in k-coalitional hedonic games},
  author={Waxman, Naftali and Kraus, Sarit and Hazon, Noam},
  journal={arXiv preprint arXiv:2001.10772},
  year={2020}
}

@article{hanaka2025maximizing,
  title={Maximizing utilitarian and egalitarian welfare of fractional hedonic games on tree-like graphs},
  author={Hanaka, Tesshu and Ikeyama, Airi and Ono, Hirotaka},
  journal={Journal of Combinatorial Optimization},
  volume={49},
  number={3},
  pages={53},
  year={2025},
  publisher={Springer}
}

@inproceedings{aziz2015welfare,
  author       = {Haris Aziz and
                  Serge Gaspers and
                  Joachim Gudmundsson and
                  Juli{\'{a}}n Mestre and
                  Hanjo T{\"{a}}ubig},
  title        = {Welfare maximization in fractional hedonic games},
  booktitle    = {Proceedings of the 24th International Joint Conference on
                  Artificial Intelligence, {IJCAI} 2015},
  pages        = {461--467},
  year         = {2015}
}

@book{garey1979computers,
  author       = {M. R. Garey and
                  David S. Johnson},
  title        = {Computers and Intractability: {A} Guide to the Theory of {NP}-Completeness},
  publisher    = {W. H. Freeman},
  year         = {1979}
}

@article{dimitrov2006simple,
  title={Simple priorities and core stability in hedonic games},
  author={Dimitrov, Dinko and Borm, Peter and Hendrickx, Ruud and Sung, Shao Chin},
  journal={Social Choice and Welfare},
  volume={26},
  number={2},
  pages={421--433},
  year={2006},
  publisher={Springer}
}

@inproceedings{flammini2025non,
  title={Non-obvious manipulability in hedonic games with friends appreciation preferences},
  author={Flammini, Michele and Fomenko, Maria and Varricchio, Giovanna},
  booktitle={Proceedings of the 24th International Conference on Autonomous Agents and Multiagent Systems, AAMAS 2025},
  pages={767--775},
  year={2025}
}

@article{dimitrov2004enemies,
author = {Dimitrov, Dinko and Sung, Shao Chin},
year = {2004},
month = {02},
pages = {},
title = {Enemies and friends in hedonic games: Individual deviations, stability and manipulation},
journal = {SSRN Electronic Journal},
doi = {10.2139/ssrn.639483}
}

@article{flammini2022strategyproof,
  author       = {Michele Flammini and
                  Bojana Kodric and
                  Giovanna Varricchio},
  title        = {Strategyproof mechanisms for friends and enemies games},
  journal      = {Artificial Intelligence},
  volume       = {302},
  pages        = {103610},
  year         = {2022}
}

@inproceedings{barrot2019unknown,
  title={Unknown agents in friends oriented hedonic games: Stability and complexity},
  author={Barrot, Nathana{\"e}l and Ota, Kazunori and Sakurai, Yuko and Yokoo, Makoto},
  booktitle={Proceedings of the 33rd AAAI Conference on Artificial Intelligence, AAAI 2019},
  pages={1756--1763},
  year={2019}
}

@article{rothe2018borda,
  title={Borda-induced hedonic games with friends, enemies, and neutral players},
  author={Rothe, J{\"o}rg and Schadrack, Hilmar and Schend, Lena},
  journal={Mathematical Social Sciences},
  volume={96},
  pages={21--36},
  year={2018},
  publisher={Elsevier}
}

@article{kerkmann2022altruistic,
  title={Altruistic hedonic games},
  author={Kerkmann, Anna Maria and Nguyen, Nhan-Tam and Rey, Anja and Rey, Lisa and Rothe, J{\"o}rg and Schend, Lena and Wiechers, Alessandra},
  journal={Journal of Artificial Intelligence Research},
  volume={75},
  pages={129--169},
  year={2022}
}

@article{klaus2023core,
  title={Core stability and strategy-proofness in hedonic coalition formation problems with friend-oriented Preferences},
  author={Klaus, Bettina and Klijn, Flip and {\"O}zbilen, Se{\c{c}}kin},
  volume       = {154},
  pages        = {16--52},
  year         = {2025},
  journal={Games and Economic Behavior},
}

@inproceedings{deligkas2025balanced,
  title={Balanced and fair partitioning of friends},
  author={Deligkas, Argyrios and Eiben, Eduard and Ioannidis, Stavros D and Knop, Du{\v{s}}an and Schierreich, {\v{S}}imon},
  booktitle={Proceedings of the 39th AAAI Conference on Artificial Intelligence, AAAI 2025},
  pages={13754--13762},
  year={2025}
}

@inproceedings{chen2023Hedonic,
  author       = {Jiehua Chen and
                  Gergely Cs{\'{a}}ji and
                  Sanjukta Roy and
                  Sofia Simola},
  title        = {Hedonic games With friends, enemies, and neutrals: Resolving open
                  questions and fine-grained complexity},
  booktitle    = {Proceedings of the 22nd International Conference on Autonomous Agents
                  and Multiagent Systems, {AAMAS} 2023},
  pages        = {251--259},
  year         = {2023}
}

@inproceedings{bullinger2025welfare,
  title={Welfare approximation in additively separable hedonic games},
  author={Bullinger, Martin and Chatziafratis, Vaggos and Shahkar, Parnian},
  booktitle={Proceedings of the 24th International Conference on Autonomous Agents and Multiagent Systems, AAMAS 2025},
  pages={418--426},
  year={2025}
}

@inproceedings{cohen2025egalitarianism,
  title={Egalitarianism in online coalition formation},
  author={Cohen, Saar and Agmon, Noa},
  booktitle={Proceedings of the 24th International Conference on Autonomous Agents and Multiagent Systems, AAMAS 2025},
  pages={2475--2477},
  year={2025}
}

@inproceedings{hanaka2019computational,
  title={Computational complexity of hedonic games on sparse graphs},
  author={Hanaka, Tesshu and Kiya, Hironori and Maei, Yasuhide and Ono, Hirotaka},
  booktitle={Proceedings of the 22nd International Conference on Principles and Practice of Multi-Agent Systems, PRIMA 2019},
  pages={576--584},
  year={2019}
}

\newpage\appendix\section{Appendix}

\subsection{Algorithm for $\FA$ With Symmetric Preferences}
\begin{algorithm}[b]
\caption{A $(2-\frac{4}{n+2})$-approximation for symmetric instances.}
\label{algo:symmetric}
\KwIn{A symmetric $\FA$ instance $\instance$ given by graph $\Gsf$}
\KwOut{A partition $\pi$}
$\agents_1 \gets \set{i\in \agents\mid f_i=1}$\\
\If{$\agents_1=\emptyset$}{\KwRet{$\weaklyC(\instance)$}}
${\mathcal R}\gets\set{j\mid j\in F_i\text{  for some } i\in\agents_1}$ \\
$\overline{\agents}_1\gets \agents_1\cup\mathcal R$\\
$\pi\gets {\emptyset}$\\
\ForEach{$r\in \mathcal R$}{
$S(r)\gets \set{r}\cup\set{i\in\agents_1\mid F_i=\set{r}}$\\
$\pi\gets \pi\cup \set{S(r)}$
}
${\mathcal U}\gets \agents \setminus \overline{\agents}_1$\\
\If{${\mathcal U}=\emptyset $}{
\KwRet{$\pi$}
}
${\mathcal L}\gets \set{i\in {\mathcal U} \mid F_i\cap {\mathcal U} =\emptyset}$\\
\If{${\mathcal L}\neq \emptyset$}{
${\mathcal U}\gets {\mathcal U}\setminus {\mathcal L}$\\
$\pi\gets \textsc{Balance}(\pi, {\mathcal L})$
}
\While{ $\exists i\in {\mathcal U}$ s.t. $\modulus{{\mathcal U}\cap F_i}< 2$ and $\modulus{{\mathcal U}}> \frac{n}{2} +1$ }{
$i\gets i\in {\mathcal U}$ s.t. $\modulus{{\mathcal U}\cap F_i}< 2$\\
\If{
$j\in {\mathcal U} \cap F_i $ s.t.\ ${\mathcal U} \cap F_j =\set{i}$
}{
${\mathcal U}\gets {\mathcal U}\setminus\set{i, j}$\\
$\pi\gets \pi \cup \set{\set{i,j}}$
}
\Else
{
${\mathcal U}\gets {\mathcal U}\setminus\set{i}$\\
$r \gets r\in \{r\in{\mathcal R}\mid S(r)\cap F_i\neq\emptyset\}$\\
$S(r) \gets S(r)\cup\set{i}$
}
}
$\pi\gets \pi\cup \set{{\mathcal U}}$\\
\KwRet{$\pi$}
\end{algorithm}

Algorithm~\ref{algo:symmetric} is a formal description of what we presented in the proof of \Cref{thm:symmetricFA}.

The procedure $\textsc{Balance}(\pi, {\mathcal L})$ takes as input a family of stars in $\Gsf$, denoted by $\pi$, and an independent set $\mathcal L$ whose nodes are connected only to the roots of stars in $\pi$. The output is a star partition where the nodes in $\mathcal L$ are placed in one of the stars so that the size of the largest star is minimized. The algorithm for $\textsc{Balance}(\pi, {\mathcal L})$ proceeds by solving a sequence of b-matching problems, as explained in the proof of  \Cref{thm:symmetricFA}.

\subsection{Tight Bounds}
The next example shows that the bound provided in \Cref{thm:apxWC} is tight.

\begin{example}\label{ex:LBweakly}
    For the case $f_{\min}=1$,
    consider an instance with $n=2k$ for some positive integer $k$, and $\agents=\set{1, \dots, n}$. Suppose that the set of friendship relationships is given by (a) $(i,i+1)\in F$ for all $i\in[n-1]$, (b) $(n,1)\in F$, and (c) $(2j, 2j-1)\in F$ for all $j\in [k]$. Note that $f_i\ge 1$ for all $i\in\agents$. 
    \weaklyC\ will put all the agents in the grand coalition, attaining an egalitarian welfare of $\frac{2}{n}$. However, creating a coalition for each pair $2j, 2j-1$, where $j\in[k]$, guarantees each agent a utility of $1$.

    For the case $f_{\min}\geq 2$, consider a graph $\Gf$ that is made of two cliques of mutual friendship relationships, of size $3$ and $n-3$, respectively, with a single directed edge from the clique of size $3$ to the clique of size $n-3$. \weaklyC\ will place all the agents in the grand coalition, attaining an egalitarian welfare of $1 + \frac{3}{n}$. However, if the two cliques are split into two disjoint, the egalitarian welfare is $2$.  
\end{example}

The next example shows that the bound provided in \Cref{thm:approxRandom} is tight.

\begin{example}
    Consider an instance $\instance$ with $n$ agents, where agents $1$ and $2$ are mutual friends, all agents $3, \dots, n-1$ are friends with each other and consider agent $n$ to be their friend, and, finally,  agent $n$ considers agents $1$ and $3$ to be her friends. We have $\opt(\instance)=1$. 
    
    Observe that \oneFalgo\ would put  $1$ and $2$ in the same coalition and place all the remaining agents in singleton coalitions, while \weaklyC\ would form the grand coalition.
    Accordingly, for this instance we have $v=1$ and $\alpha=\frac{1}{2+\frac{1}{n}}=\frac{n}{2n+1}$. Then, for $i=1, 2$ we have 
    $$
    \expectation{u_i(\randAlgo(\instance))}= \frac{n}{2n+1}\cdot \frac{2}{n} + \frac{n+1}{2n+1} = \frac{n+3}{2n+1}.
    $$
    Moreover,
    $$\expectation{u_n(\randAlgo(\instance))}= \frac{n}{2n+1}\cdot\frac{n+ 3}{n}= \frac{n+ 3}{2n+1}.
    $$
    All other agents have a higher utility in expectation. Hence, the bound is tight.
\end{example}

Finally, we show that the bound in \Cref{thm:symmetricFA} is tight.

\begin{example}
We construct a symmetric instance with $n\geq 4$ nodes, where $n=2k$ is an even number. The graph $\Gsf$ is defined as follows:
$1,2,3$ are all friends with each other, and $3, 4, \dots, n$ form a path, that is, $(i,i+1), (i+1, i)\in F^s$, for  $i=3, \dots, n-1$. Since agent $1$ only has one friend, we have $\opt\le 1$. 
Consequently, the partition $\pi=\{\{2j-1, 2j\}\mid j=1, \dots, k\}$ with $\ESW(\pi)=1$ maximizes the egalitarian welfare. However, Algorithm~\ref{algo:symmetric} initially creates $S(n-1) = \set{n-1,n}$. Then it sets $\mathcal L=\emptyset$, 
and, within the {\bf while} loop, assigns agents $n-2, n-3, n-4, \dots, k+1$ to 
be in the coalition with $n$ and $n-1$.
It then sets ${\mathcal U}_0=[k]$, i.e., it partitions the agents into two equal-sized coalitions. The egalitarian welfare of this partition is $1 - \frac{\frac{n}{2}-1}{n}= \frac{1}{2} +\frac{1}{n}$. This shows that our bound is tight. 
\end{example}

\end{document}